\documentclass[pre,twocolumn,superscriptaddress,amssymb]{revtex4}
\usepackage[dvipdfm]{graphicx}
\usepackage[english]{babel}
\usepackage{amsmath,amsfonts,amsthm}
\usepackage[dvipdfm]{graphicx}
\usepackage[utf8x]{inputenc}
\usepackage[T1]{fontenc}
\newtheorem{theo}{Theorem}
\newtheorem{lem}{Lemma}
\newtheorem{corol}{Corollary}

\begin{document}
\title{Efficient and exact sampling of simple graphs with given arbitrary degree sequence}
\author{Charo I. \surname{Del Genio}}
	\affiliation{Department of Physics, University of Houston, 617 Science \& Research 1, Houston, TX 77204-5005, USA}
	\affiliation{Texas Center for Superconductivity (TcSuH), 202 Houston Science Center, Houston, TX 77204-5002, USA}
\author{Hyunju Kim}
	\affiliation{Interdisciplinary Center for Network Science and Applications (iCeNSA), Department of Physics, University of Notre Dame, Notre Dame, IN 46556, USA}
\author{Zoltán Toroczkai}
	\affiliation{Interdisciplinary Center for Network Science and Applications (iCeNSA), Department of Physics, University of Notre Dame, Notre Dame, IN 46556, USA}
\author{Kevin E. Bassler}
	\affiliation{Department of Physics, University of Houston, 617 Science \& Research 1, Houston, TX 77204-5005, USA}
	\affiliation{Texas Center for Superconductivity (TcSuH), 202 Houston Science Center, Houston, TX 77204-5002, USA}

\begin{abstract}
Uniform sampling from graphical realizations of a given degree
sequence is a fundamental component in simulation-based measurements
of network observables, with applications ranging from epidemics, through
social networks to Internet modeling. Existing graph sampling
methods are either link-swap based (Markov-Chain Monte Carlo algorithms)
or stub-matching based (the Configuration Model). Both types are
ill-controlled, with typically unknown mixing times for link-swap
methods and uncontrolled rejections for the Configuration Model.
Here we propose an efficient, polynomial time algorithm that generates
statistically independent graph samples with a given, arbitrary,
degree sequence. The algorithm provides a weight associated with
each sample, allowing the observable to be measured either uniformly
over the graph ensemble, or, alternatively, with a desired distribution.
Unlike other algorithms, this method always produces a sample,
without back-tracking or rejections. Using a central limit
theorem-based reasoning, we argue, that for large $N$, and for degree
sequences admitting many realizations, the sample weights are expected to have
a lognormal distribution. As examples, we apply our algorithm
to generate networks with degree sequences drawn from power-law
distributions and from binomial distributions. 
\end{abstract}

\maketitle

\section{Introduction}
Network representation has become an increasingly
widespread methodology of analysis to gain insight into the
behavior of complex systems, ranging from gene regulatory
networks to human infrastructures such as the Internet, power-grids
and airline transportation, through metabolism, epidemics
and social sciences~\cite{Bar02,New03,Boc06,New06}. 
These studies are primarily data driven,
where connectivity information is collected, and the structural
properties of the resulting graphs are analyzed for modeling purposes.
However, rather frequently, full connectivity data is unavailable,
and the modeling has to resort to considerations on the
\emph{class of graphs} that obeys the available structural data.
A rather typical situation is when the only information available about the
network is the degree sequence of its nodes $\mathcal D=\left\lbrace d_0, 
d_1, \dotsc, d_{N-1}\right\rbrace$.
For example, in epidemiology studies of sexually transmitted diseases~\cite{Lil01},
anonymous surveys may only collect the \emph{number} of sexual partners
of a person in a given period of time, not their identity. Epidemiologists
are then faced with constructing a \emph{typical} contact graph having
the observed degree sequence, on which disease spread scenarios can be tested.
 Another reason for studying classes or \emph{ensembles} of graphs 
obeying constraints comes from the fact that
the network structure of many large-scale real-world systems is not the result
of a global design, but of complex dynamical processes with many
stochastic elements. Accordingly, a statistical mechanics 
approach~\cite{Bar02} can be employed to characterize the 
collective properties of the system emerging from its node level (microscopic) 
properties. In this approach, statistical ensembles of graphs are 
defined~\cite{Bia09,Bia08},
representing ``connectivity microstates'' from which macroscopic
system level properties are inferred via averaging. Here we focus 
on the degree as a node characteristic, which could represent, for example, the number 
of friends of a person, the valence of
an atom in a chemical compound, the number of clients of a router, etc. 

In spite of its practical importance, finding a method to construct
degree-based graphs in a way that allows the corresponding graph ensemble
to be properly sampled has been a long-standing open
problem in the network modeling community (references using various approaches are given below).
Here we present a solution to this problem, using a biased sampling approach.
We consider degree-based graph ensembles on two levels: 1) sequence-level, where
a specific sequence of degrees is given, and 2) distribution level, where the
sequences are themselves drawn from a given degree distribution $P\left(d\right)$.
In the remainder we will focus on the fundamental case of labeled,
undirected simple graphs. In a simple graph any link connects a single pair of
distinct nodes and self loops and multiple
links between the same pair of nodes are not allowed. Without loss of generality, consider a
sequence of $N$ positive integers $\mathcal D=\left\lbrace d_0, d_1, \dotsc, d_{N-1}\right\rbrace$,
arranged in non-increasing order: $d_0\geqslant d_1\geqslant\dotsb\geqslant d_{N-1}$.
If there is at least one simple graph $G$ with degree sequence $\mathcal D$, the
sequence $\mathcal D$ is called a \emph{graphical sequence} and we say that 
$G$ \emph{realizes} $\mathcal D$. 
Note that not every sequence of positive integers can be realized by
simple graphs. For example, there is no simple graph with degree sequence
$\left\lbrace3,2,1\right\rbrace$ or $\left\lbrace5,4,3,2,1,1\right\rbrace$, while
the sequence $\left\lbrace3,3,2,2,2\right\rbrace$ can obviously be realized by a
simple graph. In general, if a sequence is graphical, then there can be several graphs 
having the same degree sequence. Also note that given a graphical sequence, 
the careless or random placing of links between the nodes may not result in a simple graph.

Recently, a direct, swap-free method to systematically construct
all the simple graphs realizing a given graphical sequence $\mathcal D$
was presented~\cite{Kim09}.
However, in general (for exceptions see Ref.~\onlinecite{Kor76}), the number of elements of
the set $\mathcal G\left(\mathcal D\right)$ of all graphs that realize sequence $\mathcal D$, 
increases very quickly with $N$: a simple
upper bound is provided by the number of all graphs with sequence $\mathcal D$,
allowing for multiple links and loops:
$\left|\mathcal G\left(\mathcal D\right)\right|\leqslant\prod_{i=0}^{N-1}d_i!\;$.
Thus, typically, systematically constructing all graphs with a
given sequence $\mathcal D$ is practical only for short sequences, such as when determining
the structural isomers of alkanes~\cite{Kim09}. For larger sequences, and in particular
for modeling real-world complex networks, 
it becomes necessary to sample $\mathcal G\left(\mathcal D\right)$.
Accordingly, several variants based on the Markov Chain
Monte Carlo (MCMC) method were developed. 
They use link-swaps~\cite{Tay82} (``switches'') to produce 
pseudo-random samples from $\mathcal G\left(\mathcal D\right)$. Unfortunately,
most of them are based on heuristics, and apart from some
special sequences, little has been rigorously shown about the methods' mixing
time, and accordingly they are ill-controlled. The literature on such MCMC methods
is simply too extensive to be reviewed here, instead, we refer the interested reader to 
Refs~\onlinecite{Coo07,Kan99,Vig05} and the references therein.
Finally, we recall the main swap-free
method producing uniform random samples from $\mathcal G\left(\mathcal D\right)$,
namely the configuration model 
(CM)~\cite{Bol80,Ben78,Mol95,New01}. 
This method picks a pair of nodes
uniformly at random and connects them, until a rejection occurs due to a double link
or a self-loop, in which case it restarts from the very beginning.
For this reason, the CM can become very slow, as shown in the Discussion
section. The CM has inspired approximation methods as well~\cite{Bri06}
and methods that construct random graphs with given \emph{expected} degrees~\cite{Chu02}.

Here, by developing new results from the theorems in Ref.~\onlinecite{Kim09}, 
we present an efficient algorithm that solves this fundamental graph
sampling problem, and it is exact in the sense that it is not based on any heuristics. 
Given a graphical sequence, the algorithm always finishes with a simple
graph realization in polynomial time, and it is rejection free. While the samples
obtained are not uniformly generated, the algorithm also provides the exact weight 
for each sample, which can then be used to produce averages of arbitrary graph observables 
measured uniformly, or following any given distribution over $\mathcal G\left(\mathcal D\right)$.

\section{Mathematical foundations}
Before introducing the algorithm, we state some results that will be
useful later on. We begin with the Erdős-Gallai (EG) theorem~\cite{Erd60}, 
which is a fundamental result that allows us to determine whether 
a given sequence of non-negative integers, called ``degree sequence'' hereafter, is graphical.
\begin{theo}[Erdős-Gallai]\label{EG}
 A non-increasing degree sequence 
 $\mathcal{D}=\left\lbrace d_0, d_1, \dotsc, d_{N-1}\right\rbrace$
is graphical if and only if their sum is even and, for all $0 \leqslant k < N-1$:
\begin{equation}\label{EGeq}
\sum_{i=0}^kd_i\leqslant k\left( k+1\right) +\sum_{i=k+1}^{N-1}\min\left\lbrace k+1,d_i\right\rbrace \:.
\end{equation}
\end{theo}
A necessary and sufficient condition for the graphicality of a degree sequence,
which is constrained from having links between some node and a 
``forbidden set" 
of other nodes
is given by the star-constrained graphicality theorem~\cite{Kim09}. 
In this case the forbidden links are all incident on one node and thus form a
``star''. To state the theorem, we first define 
the ``leftmost adjacency set'' of a node $i$ with degree $d_i$ in a 
degree sequence $\mathcal{D}$ 
as the set consisting of the $d_i$ nodes with the largest degrees 
that are \emph{not in} the forbidden set. 
If $\mathcal{D}$ is non-increasing, then the nodes in the leftmost adjacency
set are 
the first $d_i$ 
nodes in the sequence that are not in the forbidden set. 
The forbidden set could represent nodes that are either already connected to $i$,
and thus subsequent connections to them are forbidden, or just imposed arbitrarily.
Using this definition, the theorem is:
\begin{theo}[Star-constrained graphical sequences]\label{Th6}
 Let $\mathcal{D}=\left\lbrace d_0, d_1, \dotsc, d_{N-1}\right\rbrace$ be a
non-increasing graphical degree sequence. Assume there is a set of forbidden
links incident on a node $i$. Then a simple graph avoiding the
forbidden links can be constructed if and only if a simple graph can be constructed
where $i$ is connected to all the nodes in its leftmost adjacency set.
\end{theo}
A direct consequence~\cite{Kim09} of Theorem~\ref{Th6} for the case of 
an empty forbidden set is the well-known Havel-Hakimi result~\cite{Hav55,Hak62}, 
which in turn implies:
\begin{corol}\label{FirstHappy}
 Let $\mathcal{D}=\left\lbrace d_0, d_1, \dotsc, d_{N-1}\right\rbrace$ be a
non-increasing unconstrained graphical degree sequence. Then, given any node
$i$, there is a realization of $\mathcal{D}$ that includes a link between the
first node and $i$.
\end{corol}
Another result we exploit here is Lemma 3 of Ref.~\onlinecite{Kim09}, extended 
to star-constrained sequences:
\begin{lem}\label{L3Kim}
Let $\mathcal D$ be a graphical sequence, possibly with a star constraint incident on 
node $i$.
Let $j$ and $k$ be distinct nodes not in the forbidden set and different from $i$,
such that $d_j > d_k$. Then $\mathcal{D'}=\left\lbrace d_0,\ldots,d_j-1,\ldots,d_k + 1,\ldots,d_{N-1}\right\rbrace$ is 
also a graphical sequence with the same star constraint.
\end{lem}
\begin{proof}
Let $\mathcal X_i$ denote the set of nodes forbidden to connect to node $i$. 
Since $\mathcal D$ is star-constrained graphical there is a simple
graph $G$ realizing the sequence with no connections between $i$ and $\mathcal X_i$. Since
$d_j > d_k$, there is a node $m$ to which $j$ is connected but $k$ is not. Note that 
$m$ could 
be in $\mathcal X_i \cup \left\lbrace i\right\rbrace$. Now cut the edge $(m,j)$ of $G$ creating
a stub at $m$ and another at $j$. Remove the stub at $j$ so that its degree becomes 
$d_j - 1$, and add a stub at $k$ so that its degree becoming $d_k+1$. 
Since there are no
connections in $G$ between $m$ and $k$, connect the two stubs at these nodes creating
a simple graph $G'$ thus realizing $\mathcal D'$. Clearly there are still no connections 
between $i$ and $\mathcal X_i$ in $G'$, and thus $\mathcal D'$ is also star-constrained 
graphical.
\end{proof}

Finally, using Lemma~\ref{L3Kim} and Theorem~\ref{Th6}, we prove:
\begin{theo}\label{FMT}
Let $\mathcal{D}$ be a degree sequence, possibly with a star-constraint
incident on node $i$, and let $y$ and $z$ be two nodes with degrees such
that $d_y \geqslant d_z$ that are not constrained from linking to node $i$. If the
residual degree sequence $\mathcal{D}'$ obtained from $\mathcal{D}$ by 
reducing the degrees at $i$ and $y$ by unity is not graphical, then the degree 
sequence $\mathcal{D}''$ obtained from $\mathcal{D}$ by reducing 
the degrees at $i$ and $z$ by unity is also not graphical.
\end{theo}
\begin{proof}
By definition, $d'_l=d_l$ for $l \in \left\lbrace0,\ldots,N-1\right\rbrace \setminus \left\lbrace i,y\right\rbrace$ and 
$d'_i = d_i - 1$, $d'_y = d_y - 1$; $d''_l=d_l$ for $l \in \left\lbrace0,\ldots,N-1\right\rbrace 
\setminus \left\lbrace i,z\right\rbrace$ and $d''_i = d_i - 1$, $d''_z = d_z - 1$. 
We consider $d_i \geqslant d_y$, however, the proof is not affected by this 
assumption. By assumption, $\mathcal D'$ is not graphical. 
Using proof by contradiction, assume that $\mathcal D'' = \left\lbrace\ldots,d_i - 1,\ldots, d_y,\ldots,
d_z-1, \ldots \right\rbrace$ is graphical.  Clearly, $d_y > d_z - 1$, and thus we can apply Lemma 
\ref{L3Kim} on this sequence. 
As a result, the sequence $\left\lbrace \ldots, d_i - 1, \ldots, d_y-1,\ldots, d_z-1+1,\ldots \right\rbrace$, that is 
exactly $\mathcal D'$ is graphical, a contradiction.
\end{proof}

Note that if a sequence is non-graphical, then it is not star-constrained graphical either,
and thus Theorem~\ref{FMT} is in its strongest form. 

\section{Biased sampling}
The sampling algorithm described below is ergodic in the sense that every possible
simple graph with the given finite degree sequence is generated with non-zero probability.
However, it does not generate the samples with uniform probability; the sampling is biased.
Nevertheless, the algorithm can be used to compute network observables that are 
unbiased, by appropriately weighing the averages measured
from the samples. According to a well known principle of biased sampling~\cite{New99,Coc77}, if the
relative probability of generating a particular sample $s_i$ is $\rho_{s_i}$, then
an unbiased estimator for an observable $Q$ measured from a set of $M$ randomly
generated samples $s_1, s_2, \ldots, s_M$ is the weighted average
\begin{equation}\label{two}
\langle Q\rangle=\frac{\sum_{i=1}^M w_{s_i} Q\left(s_i\right)}{\sum_{i=1}^{M}w_{s_i}}\:,
\end{equation}
where the weights are $w_{s_i}=\rho^{-1}_{s_i}$, and the denominator is a normalization
factor. The key to this method is to find the appropriate weight $w_{s_i}$ to associate
with each sample. Note that in addition to uniform sampling, it is in fact possible to
sample with any arbitrary distribution by choosing an appropriate set of sample weights.

\section{The algorithm}
Let $\mathcal D$ be a non-increasing
graphical sequence. We wish to sample the set $\mathcal G\left(\mathcal D\right)$ of graphs that realize 
this sequence. The graphs can be systematically
constructed by forming all the links involving each node. To do so, begin by
choosing the first node in the sequence as the ``hub'' node and then build the set
of the ``allowed nodes'' $\mathcal{A}=\left\lbrace a_1, a_2, \dotsc, a_k\right\rbrace$
that can be connected to it. $\mathcal{A}$ contains all the nodes that can be connected
to the hub such that if a link is placed between the hub and a node from $\mathcal{A}$, 
then a simple graph can
still be constructed, thus preserving graphicality.
Choose uniformly at random a node $a\in\mathcal{A}$, and place a link 
between $a$ and the hub. If $a$ still has ``stubs'', i.e.\ remaining links to be
placed, then add it to the set of ``forbidden nodes'' $\mathcal{X}$ that contains
all the nodes which can't be linked anymore to the hub node and which initially 
contains only the hub;
otherwise, if $a$ has no more stubs to connect, then remove it from further consideration.
Repeat the construction of $\mathcal{A}$ and link the hub with one
of its randomly chosen elements until the stubs of the hub are exhausted. 
Then remove the hub from further consideration, and repeat the whole procedure 
until all the links are made and the sample construction is complete. Each time
the procedure is repeated, the degree sequence $\mathcal{D}$ considered is 
the ``residual
degree sequence'', that is the original degree sequence reduced by the links that
have previously been made, 
and with any 
zero residual degree node removed from the sequence.
Then, choose a new hub, empty
the set
of forbidden nodes $\mathcal{X}$ and add the new hub to it.
It is convenient, but not necessary, to choose the new hub to be a node with 
maximum degree in the residual degree sequence.

The sample weights needed to obtain unbiased estimates using Eq.~\ref{two}
are the inverse relative probabilities of generating the particular samples.
If in the course of the construction of the sample $m$ different nodes $i=1,2,\ldots,m$
are chosen as the hub and they have $\bar d_i$ residual degrees when they are 
chosen, then this sample weight can be computed by first taking the product
of the sizes $k_{i_j}$ of the allowed sets $\mathcal{A}$ constructed,
then dividing this quantity by a combinatorial factor which is the product of the 
factorials of the residual degrees of each hub:
\begin{equation}\label{weight}
w=\prod_{i=1}^{m}\frac{1}{\bar d_i!}\prod_{j=1}^{\bar d_i}k_{i_j}\;.
\end{equation}
The weight accounts for the fact that at each step the hub node
has $k_{i_j}$ nodes it can be linked to, which is the size of the 
allowed set at that point, and that 
the number of equivalent ways
to connect the residual stubs of a new hub is $\bar d_i!$.
Note that it is always true that $w \geqslant 1$, with $w=1$ occurring 
for sequences for which there is only one possible graph.

\subsection{Building the allowed set}
The most difficult step in the sampling algorithm is to construct the set of
allowed nodes $\mathcal{A}$. 
In order to do so first note that Theorem~\ref{FMT} implies that if a 
non-forbidden node, that is a node not in $\mathcal X$, can be added to $\mathcal A$, 
then all non-forbidden nodes with equal or higher degree can also be added 
to $\mathcal A$.
Conversely, if it is determined that a non-forbidden node cannot 
be added to $\mathcal A$, then all nodes with equal or smaller degree 
also cannot be added to $\mathcal A$.
Therefore, referring
to the degrees of nodes that cannot be added to $\mathcal{A}$ as ``fail-degrees'',
the key to efficiently construct $\mathcal A$ is to determine the maximum fail-degree, 
if fail-degrees exist.

The first time $\mathcal{A}$ is constructed for a new hub,
according to Corollary~\ref{FirstHappy}, there is no fail-degree and $\mathcal{A}$
consists of all the other nodes.
However, constructing $\mathcal{A}$ becomes more difficult once
links have been placed from the hub to other nodes. In this case,
to find the maximum fail-degree note that
at any step during the construction
of a sample the residual sequence being used is graphical. Then, since
according to Theorem~\ref{Th6} any connection to the leftmost 
adjacency set of the hub preserves graphicality, it follows from Theorem~\ref{FMT} that
any fail-degree has to be strictly less than the degree of any node in the leftmost 
adjacency set of the hub. 

If there are non-forbidden nodes in the residual degree sequence that have 
degree less than any in its leftmost adjacency set, then the maximum 
fail-degree can be found with a procedure that exploits Theorem~\ref{Th6}. 
In particular, if the hub is connected to a node with a fail-degree,
then, 
by Theorem~\ref{Th6}, even if all the 
remaining links from the hub were connected to the remaining nodes in the leftmost adjacency set, 
the residual sequence will not be graphical. Our method
to find fail-degrees, given below, is based on this argument.

Begin by constructing a new residual sequence ${\mathcal D}'$
by temporarily assuming that 
links exist between the hub and all the nodes in its leftmost adjacency set 
\emph{except for the last one}, which has the lowest degree in the set.
The nodes temporarily linked to the hub should also be temporarily added 
to the set of forbidden nodes
$\mathcal{X}$. The nodes in 
$\mathcal{D'}$ 
should be 
ordered so that it is non-increasing, that forbidden nodes appear before non-forbidden 
nodes of the same degree, and that the hub, which now has residual degree 1, is last. 

At this point, in principle one could find the maximum fail degree by 
systematically connecting the last link of the hub 
with non-forbidden nodes of decreasing degree, and testing each time for 
graphicality using Theorem~\ref{EG}. 
If it is not graphical then the degree
of the last node connected to the hub is a fail-degree, and the node with the largest 
degree for which this is true will have the maximum fail-degree.
However, this procedure is inefficient because each time a new node is 
linked with the hub the residual sequence changes and every new sequence
must be tested for graphicality. 

A more efficient procedure to find the maximum fail-degree instead involves 
only testing the sequence 
$\mathcal{D'}$. 
To see how this can be done,
note that $\mathcal{D'}$ is a graphical sequence, by Theorem~\ref{Th6}. 
Thus, by Theorem~\ref{EG}, for all relevant values of $k$,
the left hand side of Inequality~\ref{EGeq}, $L_k$, and the right
hand side of it, $R_k$, satisfy $L_k \leqslant R_k$.
Furthermore, for the purposes of finding fail-degrees it is sufficient 
to consider linking the final stub of the hub with only the last
non-forbidden node of a given degree, if any exists. 
After any such link is made, the resulting degree-sequence ${\mathcal D}''$
will be non-increasing, and thus Theorem 1 can be applied to test it for graphicality. 
Therefore, if the degree of the node connected with the last stub of the hub is a fail-degree, 
then Inequality~\ref{EGeq} for $\mathcal{D''}$ must fail for some $k$.
For each $k$, the possible differences in $L_k$ and $R_k$ between
$\mathcal{D'}$ 
and 
$\mathcal{D''}$ 
are as follows.
$R_k$ is always reduced by 1 because the residual degree of the hub is 
reduced from 1 to 0. $R_k$ may be reduced by an another factor of 1 if the last
node 
connected to the hub,
having index $x$ and degree $d_x$,
is such that $x > k$ and $d_x < k+2$.
$L_k$ is reduced by 1 if $x\leqslant k$, otherwise it is unchanged.

Considering these conditions that can cause Inequality~\ref{EGeq} to fail
for $\mathcal{D''}$,
the set of allowed nodes $\mathcal{A}$ can be constructed
with the following algorithm that requires only
testing $\mathcal{D'}$.
Starting with
$k=0$,
compute the values of $L_k$ and $R_k$ for $\mathcal{D'}$.
There are three possible cases:
(1) $L_k=R_k$,
(2) $L_k=R_k-1$,
and
(3) $L_k\leqslant R_k-2$.
In case (1)
fail-degrees occur whenever $L_k$ is unchanged by making the final link to the hub.
Thus, the degree of the first non-forbidden node whose index is
greater than $k$ is the largest fail-degree found with this value of $k$.
In case (2)
fail-degrees occur whenever $L_k$ is unchanged and $R_k$ is reduced
by 2 by making the final link to the hub.
Thus, the degree of the first non-forbidden node whose index is
greater than $k$ 
and whose degree is less than $k+2$
is the largest fail-degree found with this value of $k$.
In case (3) no fail-degree can be found with this value of $k$.
Repeat this process sequentially increasing $k$, until all the relevant 
$k$ values have
been considered, then retain the maximum fail-degree.
It can be shown that the algorithm can be stopped either after a case
(1) occurs, or after $k=r$ where $r$ is the lowest index of any node in
$\mathcal{D'}$ with degree $d'_r < r$.
Once the maximum fail-degree is found, remove the nodes that were 
temporarily added to $\mathcal{X}$
and construct $\mathcal{A}$ by
including all non-forbidden nodes of $\mathcal D$ with a higher degree.
If no fail-degree is ever
found, then all non-forbidden nodes
of $\mathcal{D}$ are included in $\mathcal{A}$.
$\mathcal{A}$
will always include the leftmost adjacency set of the hub
and any non-forbidden nodes of equal degree.

Note that 
after a link is placed in the sample construction process, the 
residual degree sequence $\mathcal{D}$ changes, and therefore, $\mathcal{A}$ 
has to be determined every time.

\subsection{Implementing the Erd\H{o}s-Gallai test}
Finally, 
$L_k$ and $R_k$
should be calculated efficiently.
Calculating the sums that comprise them
for each new value of $k$ can be computationally
intensive, especially for long sequences. Even computing them only for as many distinct
terms as there are in the sequence, as suggested in Ref.~\onlinecite{Tri03}, can still
become slow if the degree distribution is not quickly decreasing. Instead,
it is much more efficient to use recurrence relations to calculate them.

A recurrence relation for $L_k$ is simply
\begin{equation}\label{RecS}
L_k=L_{k-1}+d_k\:,
\end{equation}
with $L_0=d_0$. 

For non-increasing degree sequences, 
define the ``crossing-index'' $x_k$ for each $k$
as the index of first node that has degree
less than $k+1$, that is for which $d_i < k+1$ for all $i \geqslant x_k$. 
If no such index exists, such as for $k=0$ since the minimum 
degree of any node in the sequence
is 1, then set $x_k=N$. 
Then, a recurrence relation for $R_k$ is 
\begin{equation}
R_k
=
R_{k-1}
+
x_k-1
-
(x_k-1-2k+d_k)\;
\Theta(k+1-x_k)
\label{RecM1}
\end{equation}
where $\Theta$ is a discrete equivalent of the Heaviside function, defined to be 1
on positive integers and 0 otherwise, and $R_0 = N-1$. 
Or, since the crossing-index can not increase with $k$, that is 
$x_{k} \leqslant x_{k-1}$
for all $k$,
a value $k^*$ will exist for which $x_k < k+1$ for all $k \geqslant k^*$,
and so Eq.~\ref{RecM1} can be written
\begin{equation}
R_k
=
\left\lbrace
\begin{array}{ll}
R_{k-1}
+
x_k-1
\qquad
&
\mbox{for} \; k < k^* \\
R_{k-1}
+
2k-d_k
&
\mbox{for} \; k \geqslant k^* \\
\end{array}
\right.
\label{RecM}
\end{equation}
Thus, 
there is no need to find $x_k$
for $k > k^*$. 

Using 
Eqs.~\ref{RecS} and~\ref{RecM}, the mechanism
of the calculation of $L_k$ and $R_k$
at sequential values of $k$
is shifted from a slow repeated
calculation of sums of many terms
to the much less computationally intensive task of
calculating the recurrence relations.
In order to perform the test efficiently,
a table of the values of crossing-index $x_k$ for each relevant $k$ can be 
created as $\mathcal{D}'$ is constructed.

\begin{figure}
 \centering
\includegraphics[width=0.40\textwidth]{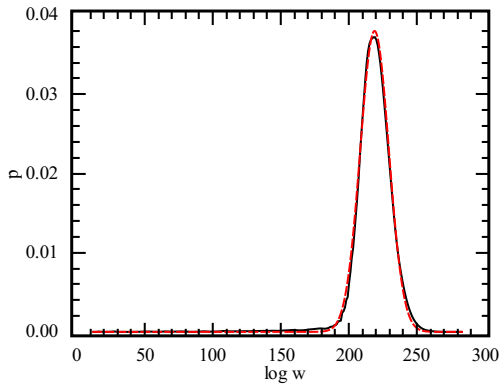}
\caption{\label{lognorm}Probability distribution $p$ of
the logarithm of weights for an ensemble of power-law sequences with $N=100$
and $\gamma=3$. The ensemble contained
$2\times10^4$ graphical sequences, and for each sequence $10^6$ graph samples were produced.
Thus, the total number of samples produced was $2\times 10^{10}$. The simulation
data is given by the solid black line and a Gaussian fit of the data is shown by the dashed red line
that nearly obscures the black line.}
\end{figure}

It should be noted that the usefulness of this method for calculating
$L_k$ and $R_k$ is broader than its use 
for calculating fail-degrees
in our sampling algorithm.
In particular, it can be used in an
Erdős-Gallai test 
to efficiently determine whether a degree-sequence is graphical.

\section{Sample weights}
As previously stated, the weight $w$ associated with a
particular sample, given by Eq.~\ref{weight}, 
is the product of the sizes $k_{i_j}$ of all the sets of allowed
nodes that have been built for each hub node $i$
divided by the product of the factorials of the initial residual degrees of
each hub node. The logarithm of this weight is
\begin{equation}
\log w=\sum_{i=1}^{m}\left[\left(\sum_{j=1}^{\bar d_i}\log k_{i_j}\right)
-\log\left(\bar d_i!\right)\right]\:.
\label{lw}
\end{equation}
Generally, degree sequences with $N\gg 1$ admit many graphical realizations.
When this is true,
each of the $m$ terms in square brackets
in Eq.~\ref{lw} are effectively random and independent,
and, by virtue of the central limit theorem,
their sum will be normally distributed. That is, the
weight $w$ of graph samples generated from a given degree sequence
with large $N$ is typically log-normally distributed.
However, 
degree sequences 
with $N \gg 1$
that have only a small number of realizations
do exist,
and
$w$ is not expected to be log-normally distributed for those sequences.

\begin{figure}
 \centering
 \includegraphics[width=0.40\textwidth]{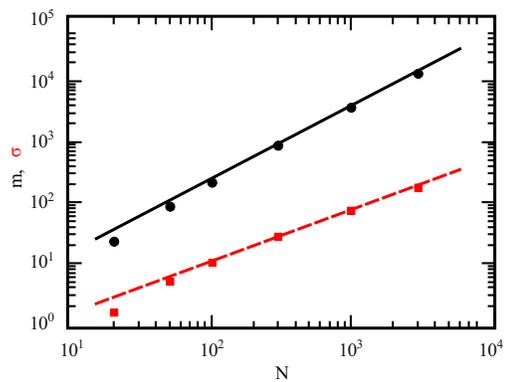}
\caption{\label{paramfit}Mean $m$ and standard deviation $\sigma$ of the distributions of
the logarithm of the weights vs.\ number of nodes $N$ of samples from an ensemble of power-law
sequences with $\gamma=3$. The black circles correspond to $m$, the red squares correspond
to $\sigma$. The error bars are smaller than the symbols. The solid black line
and the dashed red line show the outcomes of fits on the data. The linearity of the data
on a logarithmic scale indicates that the $m$ and $\sigma$ follow power-law scaling relations with
$N$: $m\sim N^\alpha$ and $\sigma\sim N^\beta$. The slopes of the fit lines are an estimate of the
value of the exponents: $\alpha=1.22042\pm0.00007$ and $\beta=0.8599\pm0.0018$.}
\end{figure}

Furthermore, one can consider not just samples of a particular 
graphical sequence, but of an ensemble of sequences.
By a similar argument to that given above for individual sequences, the 
weight $w$ of graph samples generated from an ensemble of sequences
will also typically be log-normally distributed in the limit of large $N$.
For example, consider an ensemble of sequences of randomly chosen 
power-law distributed degrees,
that is, sequences of random integers chosen from a probability distribution
$P(d) \sim d^{-\gamma}$. Hereafter, we refer to such sequences as ``power-law sequences.''
Figure~\ref{lognorm} shows the probability distribution
of the logarithm of weights for realizations of
power-law sequences with exponent $\gamma=3$ and $N=100$. 
Note that this distribution is well approximated by a Gaussian fit.

We have also studied the behavior of the mean and the standard deviation of the
probability distribution
of the logarithm of the weights of such power-law sequences as a function of $N$.
As shown in Fig.~\ref{paramfit},
they scale as a power-law. 
We have found qualitatively similar results, including power-law scaling
of the growth of the mean and variance of the distribution of $\log w$,
for binomially distributed degree sequences
that correspond to those of Erdős-Renyi random graphs with 
node connection probability
$p$ such that $pN=4$, 
and for uniformly distributed degree sequences, that is power-law sequences with
$\gamma = 0$, with 
an upper limit, or cutoff, 
of $\sqrt{N}$ for the degree of a node. 
However, 
for uniformly distributed degree sequences
without an imposed upper limit on node degrees, we find that
the sample weights are not log-normally distributed.

\section{Complexity}
In this section we discuss 
the algorithm's computational complexity. 
We first provide an upper bound on the worst case complexity, given a
degree sequence $\mathcal D$. Then, using extreme value arguments, 
we conservatively estimate the
average case complexity for degree sequences of random integers chosen from
a distribution $P(d)$. 
The latter is useful for realistically estimating the computational
costs for sampling graphs from ensembles of long sequences.

To determine an upper bound on the worst case complexity for 
constructing a sample from a given degree sequence $\mathcal D$, 
recall that the algorithm connects all the stubs of the current hub 
node before it moves on to the hub node of the new residual sequence. 
For every stub from the hub one must 
construct the allowed set $\mathcal A$.
The algorithm for constructing $\mathcal A$,
which includes constructing $\mathcal D'$, 
performing the 
$L_k$ vs $R_k$
comparisons,
and
determining the maximum fail-degree,
can be completed in $O[N-j]$ steps, where $N-j$ is the maximum possible number of nodes 
in the residual sequence after eliminating $j$ hubs from the process.
Therefore, an upper bound on the worst case complexity $C_w$ of the algorithm
given a sequence $\mathcal D$ is:
\begin{equation}
C_w \leqslant O\Big(\sum_{j} (N-j)d_j \Big) \leqslant 
O\left(N \sum_{j} d_j\right)\label{Cw}
\end{equation} 
where the sum involves at most $O(N)$ terms.
Equivalently, 
$C_w \leqslant O(NN_l)$, with $N_l$ being the number of links in the graph. 
For simple graphs, the maximum possible number of links is $O(N^2)$, and
the minimum possible number is $O(N)$.
If $N_l = O(N)$, then $C_w \leqslant O(N^2)$, and if $N_l = O(N^2)$, then
$C_w \leqslant O(N^3)$, which is an upper bound, independent of
the sequence.

\begin{figure}
 \centering
\includegraphics[width=0.40\textwidth]{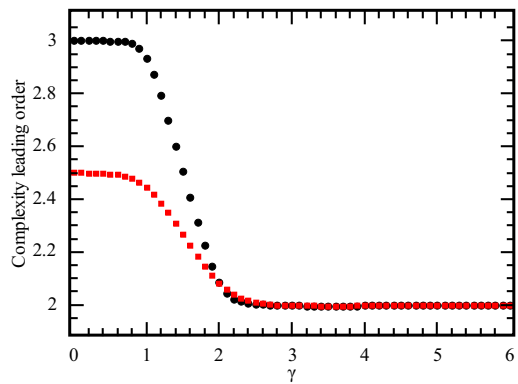}
\caption{\label{PLfits} The estimated
computational complexity of the algorithm for power-law sequences.
The leading order of the computational complexity
of the algorithm as a power of $N$, where $N$ is the number of nodes,
is plotted as a function of the degree distribution power-law exponent $\gamma$.
The black circles correspond to ensembles of sequences without cutoff, 
while the red squares correspond
to ensembles of sequences with structural cutoff in the maximum degree
of $d_{\max}=\sqrt{N}$. The fits that
yielded the data points were carried out considering sequences ranging in size 
from $N=100$ to $N=10^8$.}
\end{figure}

From Eq.~\ref{Cw}, 
the expected 
complexity for the algorithm
to construct a sample for a degree sequence of random integers chosen from
a distribution $P(d)$, normalized to unity, 
can be 
conservatively estimated as
\begin{equation}\label{compl}
C \sim O\left( \sum_{j=0}^{N-2} (N-j) \hat{d}_j \right) \:. 
\end{equation}
Here $\hat{d}_j$ is the expectation value for
the degree of the node with index $j$, 
which is the largest degree for which the expected number
of nodes with equal or larger degree is at least $j+1$. That is,
\begin{equation}\label{xhat}
\hat{d}_j
 =\max\left\lbrace d^*:\:N\sum_{d=d^*}^{d_{\max}}P(d) \geqslant j+1\right\rbrace \:.
\end{equation}
Notice that the sum in the above equation runs to the maximum allowed degree in the network
$d_{\max}$, which is nominally $N-1$, but a different value can be imposed. 
For example, in the
case of power-law sequences, the so-called structural cutoff of 
$d_{\max} \leqslant \sqrt{N}$ 
is necessary if degree correlations are to be avoided~\cite{Chu02,Bur03,Bog04}. 
However, such a cutoff needs
to be imposed only for $\gamma<3$, because
the expected maximum degree $\hat{d}_0$ in a power-law network 
grows like $N^{\frac{1}{\gamma-1}}$.
Thus, for $\gamma \geqslant 3$, $\hat{d}_{0}$ grows no faster than $\sqrt{N}$ 
and no degree correlations
exist for large $N$~\cite{Cat05}.

Given a particular form of distribution 
$P(d)$,
Eq.~\ref{compl} can be computed for different values of $N$.
Subsequent fits of the results to a power-law function
allow the order of the complexity of the algorithm to be estimated. Figure~\ref{PLfits}
shows the results of such calculations for power-law sequences with and without the structural
cutoff of $d_{\max}=\sqrt{N}$ as a function of exponent $\gamma$.
Note that, in the absence of cutoff, 
the results indicate that the order of the complexity goes to a
value of 3 for $\gamma\rightarrow0$, that is, in the limit of a uniform degree 
distribution. 
However,
if the structural cutoff is imposed the order of the 
complexity is only $2.5$ in this limit. 
Both these
results are easily verified analytically. 

We have tested the estimates shown in Fig.~\ref{PLfits}
with our implementation of the
sampling algorithm for power-law sequences with and without the structural cutoff
for certain values of $\gamma$, including 0, 2, and 3. 
This was done by measuring the actual execution times for generating samples
for different $N$ and fitting the results to a power-law function.
In every case, 
the actual order of the complexity of our implementation of the
sampling algorithm was equal to or slightly less than 
its estimated value shown in Fig.~\ref{PLfits}.

\section{Discussion}
We have solved the long standing problem of how to efficiently and accurately
sample the possible graphs of any graphical degree sequence, and of any ensemble
of degree sequences.
The algorithm we present for this purpose is ergodic and is guaranteed to
produce an independent sample in, at most, $O(N^3)$ steps. Although
the algorithm generates samples non-uniformly, and, thus, it is biased, 
the relative probability of generating each sample 
can be calculated explicitly permitting unbiased
measurements to be made. Furthermore, because the sample weights
are known explicitly, the algorithm makes it possible to sample with
any arbitrary distribution by appropriate re-weighting.

It is important to note that the sampling algorithm is guaranteed to
successfully and systematically proceed in constructing a graph.
This behavior contrasts with that of other algorithms, 
such as the configuration model (CM),
which can run into dead ends that require back-tracking or restarting,
leading to considerable losses of time and potentially introducing an uncontrollable
bias into the results. 
While there are classes of sequences 
for which it is perhaps preferable to use the CM
instead of our algorithm, in other cases its performance relative to ours can be remarkably poor.
For example, a configuration model code failed to produce even a single sample of
a uniformly distributed graphical sequence, $P(d) = const.$, with $N=100$, 
after running for more than 24 hours, while our
algorithm produced $10^4$ samples of the very same sequence in 30 seconds.
Furthermore, each sample generated by our algorithm is independent. This behavior
contrasts with that of algorithms based on MCMC methods.
Because our algorithm works for any graphical sequence and for any ensemble
of random sequences, it allows arbitrary classes of graphs to be studied.

One of the features of our algorithm that makes it efficient 
is a method of calculating the left and right sides
of the inequality in the Erdős-Gallai theorem using recursion relations.
Testing a sequence for graphicality can thus be
accomplished without requiring repeated computations
of long sums, and the method
is efficient even when the sequence is nearly non-degenerate.
The usefulness of this method is not limited to the
algorithm presented for graph sampling, but can be used anytime a fast test of the
graphicality of a sequence of integers is needed.

There are now over 6000 publications focusing on complex networks. In many of these
publications various processes, such as network growth, flow on networks, epidemics, etc.,
are studied on toy network models used as ``graph representatives'' simply because 
they have become customary to study processes on. 
These include the Erdős-Rényi random graph model, the Barabási-Albert preferential 
attachment model, the Watts-Strogatz small-world network model, random geometric graphs, etc.
However, these toy models are based on specific processes that constrain their
structure beyond their degree-distribution, which in turn might not actually correspond 
to the processes that have led to the structure of the networks investigated
with them, thus potentially introducing dangerous biases in the conclusions of these studies.
The algorithm presented here provides a way to study classes of simple graphs constrained 
solely by their degree sequence, and nothing else. However, additional constraints,
such as connectedness, or any functional
of the adjacency matrix of the graph being constructed,
can in principle be added
to the algorithm to further restrict the graph class built.

\begin{acknowledgments}
CIDG and KEB are supported by the NSF through grant DMR-0908286 and by
the Norman Hackerman Advanced Research Program through grant 95921. HK
and ZT are supported in part by the NSF BCS-0826958 and by DTRA through
HDTRA 201473-35045. The authors gratefully acknowledge Y.~Sun, B.~Danila,
M.~M.~Ercsey~Ravasz, I.~Miklós, E.~P.~Erdős and L.~A.~Székely for fruitful
comments, discussions and support.
\end{acknowledgments}

\end{document}